\documentclass[oribibl]{llncs}
 
\usepackage{microtype}
\usepackage{amssymb}
\usepackage{amsmath}
\usepackage{amsfonts}
\usepackage{graphicx}
\usepackage{xcolor}

\spnewtheorem{myclaim}{Claim}{\bfseries}{\itshape}
\spnewtheorem{fact}{Fact}{\bfseries}{\itshape}
\spnewtheorem{observation}{Observation}{\bfseries}{\itshape}

\title{Improved Approximation Algorithms for Earth-Mover
Distance in Data Streams}

\author{Arman Yousefi\and Rafail Ostrovsky}
\institute{University of California Los Angeles\\
\email{\{armany, rafail\}@cs.ucla.edu}}

\begin{document}

\maketitle

\begin{abstract}
For two multisets $S$ and $T$ of points in $[\Delta]^2$,
such that $|S| = |T|= n$,
the \textit{earth-mover distance} (EMD) between $S$ and $T$ is the minimum cost
of a perfect bipartite matching with edges between points in $S$ and $T$, i.e., 
$EMD(S,T) = \min_{\pi:S\rightarrow T}\sum_{a\in S}||a-\pi(a)||_1$,
where $\pi$ ranges over all one-to-one mappings. The sketching complexity of
approximating earth-mover distance in the two-dimensional grid is mentioned
as one of the open problems in \cite{openproblems1, openproblems2}. 
We give two algorithms for computing EMD between two multi-sets when
the number of distinct points in one set is a small value
$k=\log^{O(1)}(\Delta n)$.
Our first algorithm gives a $(1+\epsilon)$-approximation using
$O(k\epsilon^{-2}\log^{4}n)$ space and works only in the
insertion-only model.
The second algorithm gives a $O(\min(k^3,\log\Delta))$-approximation
using $O(\log^{3}\Delta\cdot\log\log\Delta\cdot\log n)$-space in the turnstile model. 
\end{abstract}

\section{Introduction}
For a metric space $X$ endowed with distance function $d_X$,
the \textit{earth-mover distance} (EMD) between two multisets $S,T\subseteq X$,
where $|S|=|T|=n$ is defined as
$EMD_X(S,T) = \min_{\pi:S\rightarrow T}\sum_{a\in S}d_X(a,\pi(a))$
where $\pi$ ranges over all bijections $\pi:A\rightarrow B$.
In this paper we mostly deal with earth-mover distance over $\ell_1$.
Thus, when the metric $X$ is $\ell_1$ we omit the subscript and 
write $EMD(S,T)=EMD_{\ell_1}(S,T)
=\min_{\pi:S\rightarrow T}\sum_{a\in S}||a-\pi(a)||_1$.
Earth-mover distance over the two dimensional plane has received significant
interest in computer vision because it is a natural measure of
similarity between images \cite{tommasi, tommasi2, grayspace, thaper}.
Each image can be viewed as a set of features and the distance is the optimal
way to match various features of images, where the
cost of such a matching corresponds to the sum of distances between the features
that were matched.
Apart from being a popular
distance measure in graphics and vision, variants of earth-mover distance known
as \textit{transportation cost} are used as LP relaxations for
classification problems such as 0-extensions and metric labeling
\cite{classification_metric, metric_labeling, charikar}.

In this paper we study two-dimensional earth-mover distance in the
streaming scenario.
In the streaming scenario for earth-mover distance, the two multisets of input
points are revealed to the algorithm as a stream of labeled points. The
algorithm maintains a short ``sketch'' of the data that can later be used to 
estimate the cost of the optimal matching. Note that the algorithm
does not produce an approximately optimal matching but only estimates its cost. 
The stream can be viewed as a
sequence of $m$ operations where each operation either adds a point to one of
the two multisets or removes a point from one of them. 
The streaming model in which insertion and deletion of
points are both allowed is referred to in the literature as the \textit{dynamic
data stream} model. It's also called the \textit{turnstile} model.
The alternative is a streaming model in which only insertions of points are
allowed, and such a data stream is referred to as \textit{insertion-only} model.

\subsubsection{Discrete geometric space}
In data stream scenario, we assume that points live in the discrete space
$\{1,\ldots,\Delta\}^2$ (denoted by $[\Delta]^2$) instead of the continuous
two-dimensional interval $[0,\Delta]^2$
where $\Delta$ is an integer upper bound on the diameter of the point set.
This is not a common assumption in computational geometry where input points
commonly have real coordinates. However, in real-life computations and in data
stream algorithms, the discrete space is a common assumption because the input
is assumed to have a finite precision.

Note that the assumption that the points of the two input multisets $S$ and $T$
live in the discrete space $[\Delta]^2$ implies that the distance between 
a point in $S$ and a point in $T$ is at least one. We assume that $S$ and $T$
are multisets, so multiple points of $S$ (or multiple points of $T$)
can share a location on the plane. However, we assume
that no point of $S$ shares location with a point of $T$. 

\subsection{Previous Results.}
Computing the earth-mover distance is a fundamental geometric problem, and
there has been extensive body of work focused on designing efficient algorithms
for this problem \cite{lawler, vai89, charikar, thaper, av04, eps_matching}.
The challenge in designing efficient streaming algorithms for earth-mover
distance is to construct and maintain a small space representation (or sketch)
of both multisets from which earth-mover distance between them can be
approximated.  In one dimension, the EMD between two multisets
can be reduced to calculating $\ell_1$ difference between two vectors
representing the point sets in $[\Delta]$. If the number of points in each
multiset is $n$, the $\ell_1$ difference between two vectors of size $\Delta$
can be approximated within a factor of $1+\epsilon$ for any $\epsilon>0$ using
$O(1/\epsilon^2\log n\log\Delta)$ (see Fact \ref{fact1}). Thus, the EMD between
two multiset of points in one dimensional space over a dynamic data stream can
be approximated within a factor of $1+\epsilon$ using
$O(1/\epsilon^2\log n\log\Delta)$ space.
This is a folklore result, and the interested reader is referred to
\cite{emdgraphmetrics} for a detailed explanation.

In \cite{ind04}, Indyk gives a $O(\log\Delta)$-approximation
algorithm for estimating the EMD between two multisets in $[\Delta]^2$ in one
pass over the data that uses $O(\log^{O(1)}(\Delta n))$ space.
His algorithm uses a probabilistic embedding of the EMD into $\ell_1$ that
has $O(\log\Delta)$ distortion \cite{thaper, charikar}.
Later, Naor and Schechtman \cite{naor} showed that any embedding of EMD into
$\ell_1$ must incur a distortion of at least $\Omega(\sqrt{\log\Delta})$, so 
it is not possible to approximate EMD over a data stream within a factor better
than $\Omega(\sqrt{\log\Delta})$ by embedding EMD into $\ell_1$. 

In \cite{const-approx} Andoni et al. gave a $O(1/c)$-approximation
algorithm for estimating EMD in the two-dimensional grid $[\Delta]^2$ using space
$O(\Delta^c\log^{O(1)}(\Delta n))$ for any $0<c<1$. 
Their algorithm uses the result of \cite{ind07}
which decomposes the cost EMD over $[\Delta]^2$ into a sum of closely related
metrics called EEMD, defined over $[\Delta^\epsilon]^2$. Each component of the
sum is a sub-matching between subsets of the two original multisets. 
In \cite{ind07} Indyk shows how to estimate the sum of sub-matchings by sampling
sub-matchings using a random distribution where the probability of choosing a
sub-matching is roughly proportional to its cost. In \cite{const-approx} the
authors show how to approximate the sum of sub-matchings over a data stream.

For earth-mover distance in high dimensions, Khot and Naor \cite{kn06} show
that any embedding of EMD over the $d$-dimensional Hamming
cube into $\ell_1$ must incur a distortion $\Omega(d)$, thus practically
losing all distance information. Andoni et al. \cite{kraught08} circumvent this
roadblock by focusing on sets with cardinalities upper-bounded by a parameter
$s$, and achieve a distortion of only $O(\log s\cdot\log (d\Delta))$.
As a result, they show a $O(\log s\log d\Delta)$-approximation
streaming algorithm that uses $O(d\log^{O(1)}(s\Delta))$ space.   

\subsection{Our Results}
In this paper we give two streaming algorithms for approximating
EMD in the two-dimensional grid $[\Delta]^2$ when the number of distinct points
in one of the multisets is polylogarithmic. This is an interesting case because
in applications the feature sets of images usually have bounded size.
A similar case for high dimensions has been studied before in \cite{kraught08},
but our constraint on the input is more relaxed than that of \cite{kraught08}
because we require a bound on the number of distinct
points in only one of the multisets while in \cite{kraught08}
they require a bound on the size of both sets.
The special case of EMD that we study
is also important because of its connections to the \textit{capacitated
$k$-median problem} with hard constraints as we explain shortly.

Our first algorithm gives a $(1+\epsilon)$-approximation for any $\epsilon > 0$
using space $O(k\epsilon^{-2}\log^{4}n)$. This algorithm uses coresets
for $k$-median problem and it works in the insertion-only model.
Our second algorithm works in the turnstile model (or dynamic geometric streams)
and it gives a weaker approximation of $O(\min(k^3, \log\Delta))$ using
$O(\log^3\Delta\cdot\log\log\Delta\cdot\log n)$ space.
Both algorithms naturally extend to work for higher
dimensions. However, the second algorithm is better suited for higher dimensions
because its memory usage does not depend exponentially on dimension $d$. 
The following table summarizes our results.

\begin{table}
\begin{center}
\begin{tabular}{|c|c|c|c|}
\hline
Algorithm & Approximation & space & Model  \\
\hline
Algorithm 1 & $1 + \epsilon$ & $O(k\epsilon^{-2}\log^{4}n)$& insertion-only\\
Algorithm 2 & $O(\min(k^3, \log\Delta))$ & $O(\log^3\Delta\cdot\log\log\Delta\cdot\log n)$& turnstile\\
\hline
\end{tabular}
\end{center}
\end{table}

\subsubsection{Connections to Capacitated $k$-median Clustering}
The non-streaming version of Capacitated $k$-median clustering
has been studied before (for example \cite{cap_kmedian1, cap_facil}),
and it is known to be harder than $k$-median clustering with no capacities.
In capacitated $k$-median clustering with uniform capacities over
a data stream, in addition to a parameter $k$ and a point set
$P\subseteq[\Delta]^2$, we are given a parameter $c \ge n/k$.
The goal is to find a set $Q\subseteq [\Delta]^2$ of size $k$ that minimizes
$\sum_{p\in P} ||p-f(p)||_2$ where $f(p)$ is one of the $k$ centers
that $p$ is assigned to, and that the number of points
assigned to each of the $k$ centers doesn't exceed its capacity $c$.

Our algorithms for earth-mover distance can be extended to algorithms for
capacitated $k$-median clustering with hard constraints.
The input point set of the capacitated $k$-median clustering can be viewed
as one of the point sets in the earth-mover distance and any set
of $k$ centers whose capacities add up to $n$ can be viewed as
the other multiset of points.
The $k$-median cost of a point set respect to a given set of centers
is the earth-mover distance between the input point set and the centers.
The streaming algorithm for earth-mover distance can be used
to keep a sketch of the input point set. At the end of the stream,
the algorithm exhaustively searches all possibilities for $k$ center points and,
for each choice of $k$ centers, all possible capacities of centers that do
not violate capacity constraints and add up to $n$.
For each possibility the algorithm approximates the
earth-mover distance between the input point set and the capacitated centers
and reports the centers with minimum value. Thus, the algorithm
exhaustively searches all $\Delta ^{O(k)}$ possibilities using small space
and returns an approximate solution
to the capacitated $k$-median problem with hard constraints.
Note that the above algorithm does not violate capacity constraints.
Thus, any of the algorithms in this paper can be turned into a streaming
algorithm for the capacitated $k$-median clustering with hard constraints.

\section{First Algorithm}
 In this section we show how to use coresets for $k$-median to give
a $(1+\epsilon)$-approximation algorithm for EMD.
 
For a point set $C$ and a point $p$, both in $\mathbb{R}^d$,
let $d(p,C)= \min_{c\in C}||p-c||_2$ denote the distance of $p$
from $C$.
For a weighted point set $P\subseteq\mathbb{R}^d$, with
an associated weight function $w: P\rightarrow\mathbb{Z}^+$
and any point set $C$ of $k$ points, we define
$Median(P,C) = \sum_{p\in P}w(p)d(p,C)$ as the the \textit{price} of
$k$-median clustering provided by $C$. In the $k$-median
problem, the goal is to find a set $C$ of at most $k$ points
in $\mathbb{R}^d$ such that $Median(P,C)$ is minimized.
We also use $Median_{opt}(P,k) = \min_{C\subseteq\mathbb{R}^d, |C|=k}Median(P,C)$
to denote the price of the optimal $k$-median clustering for $P$.

\begin{definition}[Coreset]
For a weighted point set $P\subseteq\mathbb{R}^d$, a weighted set
$P_{core}\subseteq\mathbb{R}^d$ is a $(k,\epsilon)$-coreset for the
$k$-median problem if for every set $C$ of $k$ centers:
$(1-\epsilon)\cdot Median(P,C)\le Median(P_{core},C)
\le (1+\epsilon)\cdot Median(P,C)$. 
\end{definition}

Har-Peled and Mazumdar \cite{kcoreset} prove the existence of small
coresets for the $k$-median problem and show how to construct them.
They also show how to construct and maintain coresets over data streams
using polylogarithmic space when the points are only inserted into the
stream. We use the following fact from \cite{kcoreset}.

\begin{fact}[Theorem 7.2 from \cite{kcoreset}]\label{har-peled}
 Given an insertion-only stream $P$ of $n$ points in $\mathbb{R}^d$ and
$\epsilon>0$, one can maintain a
$(k,\epsilon)$-coreset of size $O(k\epsilon^{-d}\log^2 n)$ for $k$-median.
The space used by the algorithm is $O(k\epsilon^{-d}\log^{2d+2} n)$ and the
amortized update time is $O(\log^2(k/\epsilon) + k^5)$.
\end{fact}

\subsection{Algorithm Description}
Let $S, T\subseteq [\Delta]^2$ be two multisets of points such that $|S|=|T|=n$,
where the number of distinct points in one of the sets is at most
$k = \log^{O(1)}(\Delta n)$. Assume without loss of generality that the
number of distinct points in $T$ is $k$.
The points of $S$ and $T$ are revealed to the algorithm in an insertion-only
stream, the following algorithm computes an estimate of $EMD(S,T)$ as follows.
The algorithm maintains a $(k,\epsilon)$-coreset for $k$-median
for the set $S$ using Fact \ref{har-peled}. For $d=2$ the space needed to
maintain the coreset is $O(k\epsilon^{-2}\log^{4}n)$. Let $S_{core}$ denote
the coreset of $S$. The algorithm also
keeps the entire set $T$ of points in its memory. At the end of
the stream the algorithm computes $EMD(S_{core}, T)$ using the
``Hungarian'' method \cite{lawler}.

We claim that $EMD(S_{core}, T)$ is a $(1+\epsilon)$-approximation
of $EMD(S,T)$. An important property of the coresets constructed for 
$k$-median in \cite{kcoreset} that allows us to extend the use of
coresets to earth-mover distance is the following. There is a
one-to-one correspondence between the points of $S$ and $S_{core}$.
If for any $p\in S$, $p'$ denotes the image of $p$ in $S_{core}$,
then the coreset construction guarantees that
$\sum_{p\in S}||p-p'||_2 \le \epsilon\cdot Median_{opt}(S,k)$.
Intuitively this means that each point of $S$ is snapped to
a point of $S_{core}$ such that the sum of movements of points
of $S$ is at most $\epsilon\cdot Median_{opt}(S,k)$. It's easy to
see from this property that for any set $C$ of points
$|Median(S,C)-Median(S_{core},C)|$ is at most
$\sum_{p\in S}||p-p'||_2 \le \epsilon\cdot Median_{opt}(S,k)$.

We now show how this property of the $(k,\epsilon)$-coreset
for $k$-median can be used to bound $|EMD(S_{core},T) - EMD(S,T)|$.
If for every point
$p\in S$, its image in $S_{core}$ is denoted by $p'$, we have: 

\begin{align*}
EMD(S_{core},T) - EMD(S,T) 
 &\le \sum_{p\in S}||p-p'||_1\le \sqrt{2}\cdot \sum_{p\in S}||p-p'||_2 \\
 &\le \sqrt{2}\cdot \epsilon\cdot Median_{opt}(S,k)\le O(\epsilon)\cdot EMD(S,T).
\end{align*}

The last inequality above holds because
$Median_{opt}(S,k) \le \sum_{p\in S} d(p, T)\le EMD_{\ell_2}(S,T)\le EMD_{\ell_1}(S,T)$.
Thus we have shown that $EMD(S_{core},T) \le (1+O(\epsilon))\cdot EMD(S,T)$.
We can also show that 
$EMD(S_{core},T) \ge (1-O(\epsilon))\cdot EMD(S,T)$
using a similar argument. Thus we have:

\begin{theorem}
For any $\epsilon>0$ and any two multisets $S,T\subseteq[\Delta]^2$, where
$|S|=|T|=n$, the number of distinct points in one set is bounded by $k$,
and the points are revealed to the algorithm in an insertion-only stream,
there is a one-pass streaming algorithm that approximates
$EMD(S,T)$ within a factor of $(1\pm\epsilon)$ and uses space
$O(k\epsilon^{-2}\log^{4}n)$.
\end{theorem}

\subsubsection{Why do coresets for dynamic data streams fail for EMD?}
A natural question to ask is if coresets can also be used for dynamic data
streams where insertion and deletions are both allowed.
Frahling and Sohler \cite{coreset_dynamic}
proposed a method for constructing coresets that work for dynamic data streams.
Their coreset construction is based on sampling points from the
data stream, and it works for $k$-median, but it cannot be
used for earth-mover distance.
For a given coreset $S_{core}$ of $S$, their algorithm constructs
a set $S'_{core}$ such that the point locations in $S_{core}$ and
$S'_{core}$ are the same, but the weight of every point in $S'_{core}$
differs from the corresponding point in $S_{core}$ by a factor of
at most $(1 \pm\epsilon)$.
Thus, for any set $C$ of $k$ points, $Median(S_{core},C)$ and
$Median(S'_{core},C)$ differ by at most a factor of $(1 \pm\epsilon)$
and computing the $k$-median cost for $S'_{core}$ approximates
the $k$-median cost for $S_{core}$ and $S$.
However, this argument does not work for earth-mover distance because
$EMD(S_{core},T)$ and $EMD(S'_{core},T)$ may differ 
significantly.
We mention a simple example to show that EMD is very sensitive to
the weight of points in $S_{core}$, 
let $T$ be a multiset containing two distinct points far from each other,
each with weight $n/2$.
Let also $S_{core}$ be a coreset for $S$ that contains
exactly two weighted points, each with weight $n/2$.
Each point of $S_{core}$ is at distance one to a point in $T$.
In this case $EMD(S_{core},T) = n$, but changing the weights of points in
$S_{core}$ by a factor of $(1\pm\epsilon)$ may affect the cost of
$EMD(S_{core},T)$ significantly.

\section{Second Algorithm}\label{algorithm2}
Our first algorithm gives a $(1+\epsilon)$-approximation, but it doesn't work
in dynamic geometric streams, and its space requirement is
$O(k\epsilon^{-2}\log^{4}n)$.
We next present our second algorithm that works on dynamic geometric streams
(when deletions are also allowed) and requires much less space specially for
higher dimensions, but these advantages come at the cost of a weaker
approximation ratio.
 
We start this section with some preliminaries and notations
used in our description of the second algorithm and its related proofs.
We use $E^*$ to denote the set of edges of minimum-cost bipartite perfect
matching between points of the two input multisets $S$ and $T$. 
For an edge $e$ that matches a point $p\in S$
with $q\in T$, let $||e||_1$ denote the $\ell_1$ distance between $e$'s
endpoints.
The cost of the matching $E^*$ or the earth-mover
distance between multisets $S$ and $T$ is $EMD(S,T) = \sum_{e\in E^*}||e||_1$.

For a grid over $\mathbb{R}^2$, we use a grid's \textit{cell size}
to refer to the side length of cells in the grid.
Fix a grid $G$
over $\mathbb{R}^2$ whose cell size is a positive integer. For every
multiset $S\subseteq[\Delta]^2$, we define $V_G(S)$ to be the
\textit{characteristic vector} of $S$ with respect to $G$. Each coordinate of
$V_G(S)$ corresponds to a cell of $G$ that intersects $(0,\Delta)^2$, and the
value of that coordinate is the number of points of $S$ in the corresponding 
cell.
In the context of our algorithm, we avoid having points that live on the grid 
lines so that the number of points that fall into a cell of the grid is
defined without ambiguity. Since the points have integral coordinates,
we can ensure that the points are in the interior of the grid cells
by restricting the grid lines to have
half-integral coordinates $\ldots,-\frac{3}{2}, -\frac{1}{2}, \frac{1}{2},
\frac{3}{2}, \frac{5}{2} \ldots$.

Throughout this section, we talk about grids that are shifted by some random
$2$-dimensional vectors with half-integral coordinates. 
We assume that each grid prior to shift is fixed at the origin $(0,0)$.
Thus, after shifting a grid by vector $\vec{v}=(x_0,y_0)$, the grid point
at $(0,0)$ is moved to $(x_0,y_0)$, and the rest of the grid translates
accordingly. Thus, we ensure that the lines of the shifted grid have
half-integral coordinates.

To estimate the earth-mover distance over data streams, our algorithm maintains
sketches of characteristic vectors of the two input sets
with respect to different grids. These sketches enable us to estimate
the $\ell_1$ and $\ell_0$ norms of the characteristic vectors
\footnote{The $\ell_0$ norm of a vector is
also referred to as the frequency moment $F_0$ in the literature.}. 

Let $V$ be an $N$-dimensional vector whose coordinates are values in the set
$\{1,\ldots ,M\}$. 
The $\ell_1$ norm of $V=(x_1,\ldots,x_N)$ is 
$||V||_1 = \sum_{i = 1}^N|x_i|$
and its $\ell_0$ norm is $||V||_0=|\{x_i: x_i\neq 0\}|$.
We use the following two facts from \cite{stable} and \cite{bar_yossef}
to maintain a sketch for the $\ell_1$ and $\ell_0$ norm of vector $V$ whose
coordinates are dynamically updated in a data stream. 
Each update in the stream is of the form $(i,a)$ which adds $a$ to the $i$-th
coordinate of $V$.

\begin{fact}[Theorem 2 of \cite{stable}]\label{fact1}
There is an algorithm that, for any $\epsilon, \delta >0$, estimates the
$\ell_1$ norm of $V$ up to a factor of $(1\pm\epsilon)$ with probability
$1-\delta$ and uses $O(1/\epsilon^2\cdot\log M\cdot\log(N/\delta)\cdot\log(1/\delta))$
bits of memory.
\end{fact}

\begin{fact}[Theorem 1 of \cite{bar_yossef}]\label{fact2}
There is an algorithm that, for any $\epsilon,\delta >0$, estimates the $\ell_0$
norm of $V$ up to a factor of $(1\pm\epsilon)$ with probability $1-\delta$
using $O(1/\epsilon^2\cdot \log N\cdot \log(1/\delta))$ bits of memory.
\end{fact}

\subsubsection{Our Technique}
Our second algorithm is a modification of the idea in \cite{ind04}
which uses an embedding of EMD into $\ell_1$ that has a distortion of
$O(\log\Delta)$ \cite{thaper, charikar}.
However, since any embedding of EMD into $\ell_1$ must
incur $\Omega(\sqrt{\log\Delta})$ distortion \cite{naor}, we need additional
ideas to obtain a better approximation ratio.

The algorithm of \cite{ind04} uses nested grids
$G_i, i = 0, \ldots, \log\Delta$ over $\mathbb{R}^2$ where the cell size of
grid $G_i$ is $2^i$, and a cell in $G_i$ contains $4$ cells in
$G_{i-1}$. The nested grids are shifted by a vector chosen uniformly at random.
The multiset $S$ is mapped into 
$f(S) = (V_{G_0}(S), 2V_{G_1}(S),\ldots,2^iV_{G_i}(S),\ldots,2^{\log\Delta}V_{G_{\log\Delta}}(S))$
in the $\ell_1$ space where
$V_{G_i}(S)$ denotes the characteristic vector of multiset $S$ with respect
to grid $G_i$. In other words, $f(S)$ is obtained by concatenating vectors
$V_{G_0}(S), 2V_{G_1}(S),\ldots,2^{\log\Delta}V_{G_{\log\Delta}}(S)$.
Similarly, multiset $T$ is mapped into $f(T)$, and 
to estimate $EMD(S,T)$, the value of $||f(S) - f(T)||_1$ is computed.
The distortion of the above embedding is $O(\log\Delta)$, so 
the above algorithm gives a $O(\log\Delta)$-approximation streaming algorithm
for computing $EMD(S,T)$.

Instead of using one grid per level, our algorithm uses
$O(\log\Delta)$ randomly shifted grids at each level
$i = 0, \ldots, \log\Delta$. At each level our algorithm maintains the
$\ell_1$ norm of the difference of characteristic vectors of $S$ and $T$ with
respect to every grid at that level. At the end of the stream, we choose
the grid with minimum $\ell_1$ difference at each level and compute our estimate.
This way our algorithm circumvents $\Omega(\sqrt{\log\Delta})$ lower bound
on the distortion of embedding EMD into $\ell_1$ \cite{naor}.
The proof that the above modification gives a better approximation ratio is the
main technical part of this section.

\subsection{Algorithm Description}
For every $i=0,\ldots,\log\Delta$, our algorithm builds 
$2\log\Delta$ grids over $\mathbb{R}^2$ with cells of size $2^i$
that are randomly and independently shifted. 
As the points in the stream arrive, the algorithm maintains a sketch for
the $\ell_1$ norm of the difference of characteristic vectors of $S$ and $T$
with respect to every grid. At the end of the stream the algorithm chooses, for
each level, the grid with minimum $\ell_1$ norm and reports
$Z = \frac{{\widehat{k}}^2}{2}\sum_{i = 0}^{\log\Delta}2^i\widehat{C}_i$
as the estimate of $EMD(S,T)$ where $\widehat{C}_i$ is the estimate of
the minimum $\ell_1$ norm at level $i$ and $\widehat{k}$ is an estimate
of the minimum of the number of distinct points in $S$ and $T$.
Thus, our algorithms maintains the following data structures:
\begin{enumerate}
 \item For each $i=0,\ldots,\log\Delta$ and each
$j = 1, \ldots, 2\log\Delta$, let $G_i^j$ be a grid of cell size
$2^i$ that is shifted by a vector chosen independently and
uniformly at random
(recall that the coordinates of the shift vector are half-integral).
The algorithm maintains a sketch of vector
$V_{G_i^j}(S)-V_{G_i^j}(T)$ under addition and deletion of points from
$S$ and $T$ to estimate its $\ell_1$ norm,
$||V_{G_i^j}(S)-V_{G_i^j}(T)||_1$, at the 
end of the stream. This can be done using Fact \ref{fact1}.
 \item The algorithm also maintains a sketch to determine the number of distinct
points in $S$ and $T$. This can be done using Fact \ref{fact2}
to estimate the $\ell_0$ norm of $S$ and $T$ with respect to any grid
at level $0$.
Note that all random shift vectors result in the same grid $G_0$ at level $0$, 
and there is at most one distinct point of $S$ (or $T$) in each cell of
$G_0$, so the $\ell_0$ norm of $V_{G_0}(S)$ (or $V_{G_0}(T)$) is the number of
distinct points of multiset $S$ (or $T$).  
\end{enumerate}

Let $\widehat{k}$ be minimum of the two estimates for the $\ell_0$ norms of
$V_{G_0}(S)$ and $V_{G_0}(T)$. Then, $\widehat{k}$ estimates $k$,
the minimum of the number of distinct points in $S$ and $T$. 
We define $C_i^j=||V_{G_i^j}(S)-V_{G_i^j}(T)||_1$, and $C_i =\min_j C_i^j$.
Let $\widehat{C}_i^j$ denote the algorithm's estimate of $C_i^j$ for all $i,j$.
At every level $i$ the algorithm chooses the grid $G_i^j$ that minimizes
$\widehat{C}_i^j$ and lets $\widehat{C}_i = \min_j{\widehat{C}_i^j}$.
The output of the algorithm is 
\begin{equation}
 Z = \frac{\widehat{k}^2}{2}\sum_{i = 0}^{\log\Delta}2^i\cdot \widehat{C}_i
\end{equation}
This concludes our description of the algorithm.

\subsubsection{Space usage}
The above algorithm uses $O(\log\Delta)$ grids at each level
$i = 0,\ldots,\log\Delta$, and maintains the $\ell_1$ norm of the difference of
characteristic vectors of $S$ and $T$ with respect to each grid. Each vector
$V_{G_i^j}(S)-V_{G_i^j}(T)$ has at most $O(\Delta^2)$ coordinates and each
coordinate is in $\{0,1,\ldots,n\}$. By Fact \ref{fact1}, the sketch to maintain
$C_i^j = ||V_{G_i^j}(S)-V_{G_i^j}(T)||_1$ requires $O(\log n\cdot\log\frac{\Delta}{\delta'}\cdot\log\frac{1}{\delta'})$ bits of
storage where $\delta'$ is the probability of error in estimating the norm
with respect to each grid. If we want the total error probability in estimating
all $C_i^j$ to be bounded by $\delta$, we need to set $\delta' = \frac{\delta}{2\log^2\Delta}$.
With this value of $\delta'$ the space needed to maintain each $C_i^j$ is
$O(\log n\cdot\log\frac{\Delta}{\delta}\cdot\log\log\Delta\cdot\log\frac{1}{\delta})$
and the total space used to maintain the $\ell_1$ norm of these vectors is
$O(\log^3\Delta\cdot\log\log\Delta\cdot\log n)$. 
Also the space needed to maintain the number of
distinct point in $S$ and $T$ is $O(\log\Delta)$
(by Fact \ref{fact2} from \cite{bar_yossef}). 
Thus the total space used by the algorithm is still
$O(\log^3\Delta\cdot\log\log\Delta\cdot\log n)$.

To show that the estimate $Z$ returned by the algorithm approximates $EMD(S,T)$,
we will prove upper and lower bounds on the value of $Z$ in the next
section. 

\subsection{Bounding the Cost}\label{cost}
In this section we prove upper and lower bounds on the cost of the estimate
returned by the algorithm.
By Fact \ref{fact1} and Fact \ref{fact2}, the values of $k$ and
$||V_{G_i^j}(S)-V_{G_i^j}(T)||_1$ for all $i,j$ can be estimated within a factor
of $1\pm\epsilon$ for any parameter $\epsilon > 0$.
This increases the space usage by a multiplicative factor of $O(1/\epsilon^2)$
which is ignored as we take $\epsilon$ to be some small constant.
If $\widehat{k} = (1\pm\epsilon)k$ and $\widehat{C}_i^j = (1\pm\epsilon)C_i^j$,
then:
\begin{equation}
Z = \frac{\widehat{k}^2}{2}\sum_{i = 0}^{\log\Delta}2^i\cdot \widehat{C}_i 
= (1\pm\epsilon)^3\cdot \frac{k^2}{2}\sum_{i = 0}^{\log\Delta}2^i\cdot C_i
\end{equation}

For fixed $\epsilon > 0$, the factor $(1\pm\epsilon)^3$ is a small constant.
Thus, it suffices to prove upper and lower bounds on the value of
$\frac{k^2}{2}\sum_{i = 0}^{\log\Delta}2^i\cdot C_i$ instead of
$Z = \frac{\widehat{k}^2}{2}\sum_{i = 0}^{\log\Delta}2^i\cdot \widehat{C}_i$
returned by the algorithm.
In fact, to further simplify the exposition, we prove our bounds on the value of 
$Y=\frac{{1}}{2}\sum_{i = 0}^{\log\Delta}2^i\cdot C_i$ which is scaled
by a factor of $1/k^2$. Specifically we show that 
$\Omega(1/k^2)\cdot EMD(S,T)\le Y\le O(k)\cdot EMD(S,T)$ with very high probability.
In the next lemma we show a high probability upper bound on $Y$.

\begin{lemma}[Upper bound]\label{upperbound}
With high probability, the value
$Y = \displaystyle\frac{1}{2}\sum_{i = 0}^{\log\Delta}2^i\cdot C_i$
is at most $O(k)\cdot EMD(S,T)$.
\end{lemma}
\begin{proof}
Recall that $E^*$ denotes the set of edges of the optimal matching
between points of $S$ and $T$. We say an edge $e\in E^*$ \textit{crosses} a
grid $G$ if the two endpoints of $e$ fall in different cells of $G$.

\begin{definition}[Good Grid]
A grid $G_i^j$ at level-$i$ is a \textit{good} grid if it is not
crossed by any edge $e\in E^*$ whose $\ell_1$ norm is less than
$\frac{1}{8k}$-fraction of cell
size of $G_i^j$ (i.e. $2^{i-3}/k$).
\end{definition}

To bound $Y$ in terms of $EMD(S,T)$, we show that with very high probability
at every level one of the $2\log\Delta$ randomly shifted grids is
a good grid. If we consider the set of such good grids, one per level $i$,
then every $e\in E^*$ only crosses grids whose cell size is at
most $8k||e||_1$. This allows us to charge the length of each edge $e\in E^*$ to
the grids that it crosses at different levels.

Let $E_i^j$ be the event that $G_i^j$ (i.e. the $j$-th randomly shifted grid
at level $i$) is a good grid.
The following claim states that with very high probability, there is
a good grid $G_i^j$ at every level $i$.

\begin{myclaim}\label{crossing}
$\displaystyle{\Pr[~\forall i~\exists j~ E_i^j~] \ge 1 - \frac{\log\Delta}{\Delta^2}}$.
\end{myclaim}
\begin{proof}
We assume without loss of generality that $k$ is the number of distinct points
in $T$. Each edge in the optimal matching connects one of these $k$ points to a
point in $S$. 
Any edge $e\in E^*$ where $||e||_1< 2^{i-3}/k$ connects a point $p\in T$ to a
point in $S$ which is in
a square of side length $2^{i-2}/k$ centered at $p$.
Grid $G_i^j$ is shifted by a random vector, and it intersects one of the edges
whose $\ell_1$ norm is $<2^{i-3}/k$ only if it intersects one of the $k$
squares of side length $2^{i-2}/k$ centered at points in $T$. 
The cell size of grid $G_i^j$ is $2^i$, and the side length of each square is
$2^{i-2}/k$, so the probability that a square is intersected by a line of
grid $G_i^j$ is $\le 1/2k$. By union bound the probability that
any of the $k$ squares intersect a line of grid $G_i^j$ is at most $1/2$. 
This also bounds the probability that grid $G_i^j$ is crossed by an edge of
length $<2^{i-3}/k$. Thus, the probability over random shift vectors that
grid $G_i^j$ is not a good grid is at most $1/2$.
There are $2\log\Delta$ shift vectors at level $i$, and by independence of
shift vectors the probability that all grids $G_i^j$ at level
$i$ are not good is at most 
$\big(\frac{1}{2}\big)^{2\log\Delta} = \frac{1}{\Delta^2}$.
The claim is proved by applying the union bound for all $\log\Delta$ levels.
\end{proof}

We next show how to bound $Y$ from above using Claim \ref{crossing}.
Let's assume that there is a good grid at each level $i$ denoted by $G_i^*$,
then:
\begin{equation}\label{eq1}
Y = \frac{1}{2}\sum_{i=0}^{\log\Delta} C_i\cdot2^i
  \le \frac{1}{2}\sum_{i=0}^{\log\Delta}||V_{G_i^*}(S)-V_{G_i^*}(T)||_1\cdot2^i
\end{equation}

It's easy to see that for every grid $G$: 
$||V_{G}(S)-V_{G}(T)||_1 \le 2|\{e\in E^*: e~\text{crosses}~G\}|$ for the
following reason. For every square $c\in G$, let $n_c(S)$ and $n_c(T)$ be the
the number of points of multiset $S$ and $T$ in square $c$ respectively. 
Then in every cell $c\in G$, $|n_c(S)-n_c(T)|$ is the minimum number of points
in cell $c$ that cannot be matched with a point in that square. Thus, in any
matching the total number of points that are not matched within their square is
at least $\sum_{c\in G}|n_c(S)-n_c(T)| = ||V_{G}(S)-V_{G}(T)||_1$. Each
point that is not matched within its square is an endpoint of an edge that
crosses grid $G$. Thus the number of edges in any matching that cross $G$ is
at least $||V_{G}(S)-V_{G}(T)||_1/2$. This combined with (\ref{eq1}) implies 
that:
\begin{equation}\label{eq2}
Y\le \sum_{i=0}^{\log\Delta} \Big|\{e\in E^*: e~\text{crosses}~G_i^*\}\Big|\cdot2^i
\end{equation}
Since $G_i^*$ is a good grid, the above implies that:
\begin{align*}
Y &\le \sum_{i=0}^{\log\Delta} \Big|\{e\in E^*: e~\text{crosses}~G_i^*\}\Big|\cdot2^i
  \le \sum_{i=0}^{\log\Delta} \Big|\{e\in E^*: ||e||_1 > \frac{1}{k}2^{i-3}\}\Big|\cdot2^i \\
  &\le \sum_{i=0}^{\log\Delta} 2^i\sum_{e\in E^*:||e||_1 > 2^{i-3}/k}1
  \le \sum_{e\in E^*}\sum_{i:2^{i-3}<k||e||_1} 2^i
  =  \sum_{e\in E^*}\sum_{i=0}^{3 + \log k||e||_1} 2^i\\
  &\le  \sum_{e\in E^*} 2^{4 + \log k||e||_1} = \sum_{e\in E^*} 16k||e||_1
  \le 16k\cdot EMD(S,T)
\end{align*}
Thus, $Y \le 16k\cdot EMD(S,T)$
if there is a good grid at each level $i$ which
happens with probability $1-\frac{\log\Delta}{\Delta^2}$ by Claim \ref{crossing}.
Hence, $\Pr[Y \le 16k\cdot EMD(S,T)] \ge 1-o(1)$.
\end{proof}
Using the above result, we can also show an upper bound on the expected
value of $Y$, but we omit the straightforward details. Our next lemma
establishes the lower bound on the value returned by the algorithm.
\begin{lemma}[Lower bound]\label{lowerbound}
The value $Y = \frac{1}{2}\sum_{i = 0}^{\log\Delta}2^i\cdot C_i$
is at least $\displaystyle\Omega(\frac{1}{k^2})\cdot EMD(S,T)$.
\end{lemma}
\begin{proof}
The idea of the lower bound is to charge the cost of each edge $e\in E^*$ to
the grid levels whose cell size is at most $||e||_1$. Thus at each level
only edges whose $\ell_1$ norm is at least the cell size of
that level contribute to the cost.
Then, at each level $i$ we bound from above the total number of
edges that contribute to that level in terms of
$C_{i'}=\min_j||V_{G^j_{i'}}(S) - V_{G^j_{i'}}(T)||_1$ where $i'$ is
\textit{a few} levels below $i$. Therefore, we can bound $EMD(S,T)$
from above in terms of $Y$.

For any $i$, we use $E^*_i$ to denote $\{e \in E^* : ||e||_1 \ge 2^i \}$. 
Note that $|E^*_i|\le n$ for all $i$ because there are a total of
$n$ edges in $E^*$.
We have:
\begin{align}
 EMD(S,T) &= \sum_{e\in E^*}||e||_1 \le \sum_{e\in E^*}\sum_{i = 0}^{\log||e||_1}2^i
          \le \sum_{i=0}^{\log 2\Delta}\sum_{e\in E^*: ||e||_1\ge 2^i} 2^i\nonumber\\
          & = \sum_{i=0}^{\log 2\Delta}\big|E^*_i\big|\cdot2^i
          \le \Bigg(\sum_{i= 0}^{\log k}n\cdot2^i + \sum_{i=\log 2k}^{\log 2\Delta}\big|E^*_i\big|\cdot2^i\Bigg)&\nonumber\\
          &\le 2nk + \sum_{i=\log 2k}^{\log 2\Delta}\big|E^*_i\big|\cdot2^i
          = 2nk + \sum_{i=0}^{\log 2\Delta-\log 2k}\big|E^*_{i+\log 2k}\big|\cdot 2^{i+\log 2k}\nonumber\\
          &= 2nk + 2k\sum_{i=0}^{\log \Delta-\log k}\Big|\{e\in E^*: ||e||_1\ge 2k\cdot 2^{i}\}\Big|\cdot2^{i}\label{eq3}.
\end{align}

The main tool in the proof of the lemma is the following claim which lower
bounds $C_{i} = \min_j||V_{G^j_{i}}(S) - V_{G^j_{i}}(T)||_1$
by the number of edges in the optimal matching $E^*$ whose length is at least
$k2^{i+1}$.
 
\begin{myclaim}\label{cl2}
 For all $i = 0, \ldots, \log\Delta$:
$\Big|\{e\in E^*: ||e||_1\ge k\cdot2^{i+1}\}\Big| \le \frac{k}{2}\cdot C_{i}$.
\end{myclaim}
The idea of the proof is to view grid $G$ at level $i$ as a graph where the
grid cells are vertices of the graph and the edges crossing the grid are
directed edges of the graph. We then show how to decompose the edges of the
graph into a set of paths of length $\le k$ where the start and end vertex of
each path contribute two to the value of $C_i$. Thus the total number of
such paths is at most $C_i/2$ and the total number of edges whose $\ell_1$
norm is $\ge k\cdot2^{i+1}$ is at most $kC_i/2$. The detailed proof appears in
the appendix.

We next show how to use the above claim to prove Lemma \ref{lowerbound}.
From Inequality (\ref{eq3}) above we have:
\begin{align}
 EMD(S,T) &\le 2nk + 2k\sum_{i=0}^{\log \Delta-\log k}\Big|\{e\in E^*: ||e||_1\ge k\cdot 2^{i+1}\}\Big|\cdot2^{i}\nonumber\\
          &\le 2nk + 2k^2\cdot\frac{1}{2}\sum_{i=0}^{\log \Delta-\log k}C_{i}2^i&\text{(by Claim \ref{cl2})}\nonumber\\
          &= O(k^2)\cdot Y \nonumber
\end{align}
The last inequality holds because $n= C_0/2 \le Y$.
This completes the proof of Lemma \ref{lowerbound}.
\end{proof}

Lemma \ref{upperbound} and \ref{lowerbound} together imply that our algorithm
gives a $O(k^3)$-approximation of the cost of EMD.
To ensure approximation ratio of $O(\min(k^3, \log\Delta))$, 
the algorithm holds an additional
data structure to maintain the sketch used by 
$O(\log\Delta)$-approximation algorithm of \cite{ind04}. 
The two algorithms
maintain their own sketches and at the end of the stream, each
algorithm computes its estimate of EMD using its sketch and the minimum
of the two estimates is returned. Clearly this estimate is within 
$O(\min(k^3, \log\Delta))$ factor of the cost of EMD.
Thus, we have the following:

\begin{theorem}
There is a $O(\min(k^3, \log\Delta))$-approximation that uses 
$O(\log^{O(1)}\Delta n)$ space to estimate the                                     
earth mover distance between two multiset $S,T\subseteq[\Delta]^2$ given over
a dynamic data stream, where $|S| = |T| = n$ and minimum of the number of
distinct points in $S$ and number of distinct points in $T$ is bounded by $k$. 
\end{theorem}

\section{Conclusion}\label{conclude}
We have obtained two approximation algorithm for earth-mover distance
between two multisets of points in $[\Delta]^2$ when the number of
distinct points in one set is small. Both algorithms use polylogarithmic space.
Our algorithms can be extended to give streaming algorithms for
capacitated $k$-median clustering with hard constraints.
We conclude with some natural open questions: 1) Is there a $O(1)$-approximation
algorithm for EMD with no constraints on the input size using only polylogarithmic space?
2) Can one prove a lower bound on the best approximation possible for EMD
in polylogarithmic space?
3) Are there better streaming algorithms for the capacitated $k$-median
with hard constraints? 
\bibliographystyle{alpha}
\bibliography{template}

\appendix
\section{Proof of Claim \ref{cl2}}
\begin{proof}
We prove that for \textit{any} grid $G^j_i$ at level $i$, the number of
edges of $E^*$ with $\ell_1$ norm at least $k2^{i+1}$ is at most $kC_i^j/2$,
where $C_i^j = ||V_{G_i^j}(S)-V_{G_i^j}(T)||_1$.

For grid $G_i^j$, we define the directed multigraph $\Gamma$ as follows. Every
cell $c$ of the grid that contains a point from $S$ or $T$ corresponds to a
vertex of $\Gamma$.
Each edge of $\Gamma$ corresponds to an edge of the matching $E^*$ that crosses
the grid. Recall that an edge crosses the grid if its endpoints are in different
cells. Each edge in $\Gamma$ has a \textit{length} which is equal to the
$\ell_1$ norm of the corresponding edge in the matching $E^*$.
Edges of $\Gamma$ are directed from points in $S$ to points in
$T$. Note that $S$ and $T$ are multisets, and there might be multiple copies of
a point in $S$ or in $T$.
Thus, if $p$ copies of a point in $S$ are matched to
$p$ copies of a point in $T$, there are $p$ copies of the same edge in the
matching $E^*$, and if the endpoints of the edge are in different cells,
there are $p$ edges between the corresponding vertices in $\Gamma$. Hence,
$\Gamma$ is in general a multigraph.

A few observations about graph $\Gamma$ are in order:
\begin{observation}\label{obs1}
 The length of every simple cycle (or every simple path) in graph $\Gamma$
is at most $k$.
\end{observation}
This holds because $k$ is an upper bound on the number distinct points in $T$
and also the number of vertices of $\Gamma$ with positive indegree.
\begin{observation}\label{obs2}
No cycle in $\Gamma$ contains an edge of length $\ge k\cdot2^{i+1}$. 
\end{observation}
If such a cycle exists, there are
matched pairs 
$(s_1,t_1),(s_2,t_2), \ldots, (s_\kappa, t_\kappa)\in E^*$ such that
for every $r\in[\kappa]$, points $t_r$ and $s_{r+1}$
(we define $\kappa + 1 = 1$) are in the same cell in grid $G_i^j$
and the length of one of the edges is $\ge k2^{i+1}$
(see Figure \ref{cycle}).
Thus, the matching $E^*$ costs at least $k2^{i+1}$.

\begin{figure}
\begin{center}
\scalebox{.7}
{
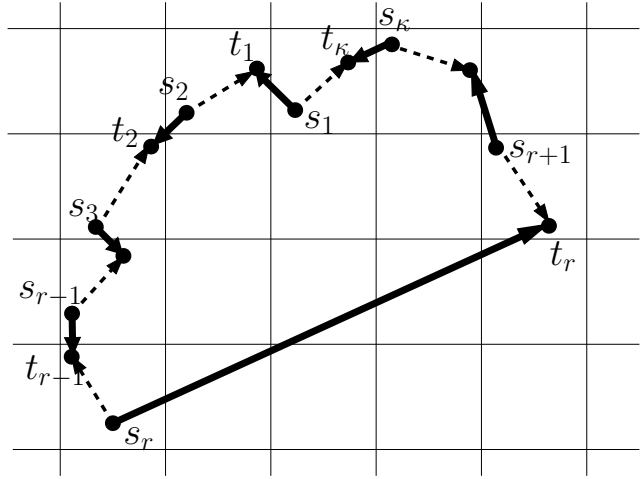
}
\caption{The solid edges show the matching of the pairs
$(s_1,t_1),(s_2,t_2),\ldots$. These edges form a cycle
in the corresponding graph $\Gamma$. If there is an edge $(s_r,t_r)$ with
length $\ge k2^{i+1}$, there is an alternative matching
$(t_1,s_2),(t_2, s_3),\ldots$ (shown in dotted lines) with lesser cost.
\label{cycle}}
\end{center}
\end{figure}
However, there is an alternative matching
$(t_1,s_2),(t_2, s_3),\ldots,(t_\kappa, s_1)$
that costs at most $k2^{i+1}$ because the are at most $k$
pairs of points (by Observations \ref{obs1}),
and each pair is in the same cell of grid $G_i^j$.
This contradicts the optimality of $E^*$.

The next observation is a connection between the number of points of $S$ and $T$
in a cell of the grid, and the indegree and outdegree of the corresponding cell
in graph $\Gamma$. Let $c$ be any cell in $G_i^j$ and $v_c$ be the corresponding
vertex in graph $\Gamma$. Let $\deg^+(v_c)$ and $\deg^-(v_c)$ denote the
outdegree and indegree of $v_c$ respectively, and let $n_c(S)$ and $n_c(T)$
denote the number of points of $S$ and $T$ in the cell $c$.
\begin{observation}
For every cell $c\in G_i^j$ and corresponding vertex $v_c$ in $\Gamma$:

$\deg^+(v_c)-\deg^-(v_c)=n_c(S)-n_c(T)$.
\end{observation}
For every edge out of $v_c$,
there should be a point of $S$ in cell $c$, and for every edge into $v_c$ there
should be a point of $T$ in cell $c$. The rest of the points of $S$ and $T$
in cell $c$ that are not an endpoint of an edge of $\Gamma$ are matched within
the cell. The number of points of $S$ that are matched within $c$ should be
equal to the number of points of $T$ that are matched within $c$, so these
points don't contribute any value to $n_c(S)-n_c(T)$.
Hence, $\deg^+(v_c)-\deg^-(v_c)=n_c(S)-n_c(T)$. 

By summing over all cells $c$ in the grid, this observation implies that:
\begin{equation}\label{eq4}
 \sum_{v\in{V(\Gamma)}}|\deg^+(v)-\deg^-(v)| = \sum_{c\in G_i^j}|n_c(S)-n_c(T)| = ||V_{G_i^j}(S)-V_{G_i^j}(T)||_1 = C_i^j
\end{equation}

We are now ready to prove claim \ref{cl2}. The idea is
to decompose the edges of graph $\Gamma$ into a
set of paths where each path contains at least one of the edges of length
$\ge k2^{i+1}$.
We show that the graph $\Gamma$ can be decomposed into at most
$C_i^j/2$ such paths and each path contains at most $k$
edges of length $\ge k2^{i+1}$.

The decomposition works in a natural way as follows.
Let $e=(u,t)$ be any edge such that $||e|| > k2^{i+1}$.
We show how to construct a path that contains $e$.
If $\deg^-(t) > \deg^+(t)$, $t$ is the end of the path. Otherwise,
$\deg^-(t) \le \deg^+(t)$ and there is an edge $(t,w)$ going out of $t$. By
the same argument either $\deg^-(w) > \deg^+(w)$ or there is an edge going
out of $w$. This process can be repeated until a vertex $z$ is reached such
that $\deg^-(z) > \deg^+(z)$. We mark $z$ to be the end of the path.
The original edge $e=(u,t)$ can also be traced in the opposite direction to
reach a vertex $a$ such that $\deg^+(a) > \deg^-(a)$.
Note that by Observation \ref{obs2}
edge $e$ is not in any cycle, so the start and end vertex of the path can't
be the same vertex. Removing all the edges of this path from graph $\Gamma$
reduces $\sum_{v\in{V(\Gamma)}}|\deg^+(v)-\deg^-(v)|$ by two because
the quantity $|\deg^+(v)-\deg^-(v)|$ is reduced by one 
for the start and end vertices of the path, and for all
other vertices this quantity is unchanged.

After removing the edges of this path from graph $\Gamma$,
the remaining graph may still contain
edges of length $\ge k2^{i+1}$. We can choose any one of these edges and
repeat the above process to find another path and remove it from the graph.
This process can be repeated until there are no edges of length
$\ge k\cdot2^{i+1}$ in the graph. Each time a path is extracted the quantity
$\sum_{v\in{V(\Gamma)}}|\deg^+(v)-\deg^-(v)|$ reduces by two. Thus, the total
number of such paths is at most
$\frac{1}{2}\sum_{v\in{V(\Gamma)}}|\deg^+(v)-\deg^-(v)|$
which equals $C_i^j/2$ by Equation (\ref{eq4}).
Each path contains at most $k$ edges of length $\ge k\cdot2^{i+1}$
by Observation \ref{obs1}. Thus the total number of such edges is
bounded by $kC_i^j/2$. 
\end{proof}

\end{document}